\documentclass[twocolumn]{article}
\usepackage{cite}
\usepackage{amsmath}
\usepackage[margin=0.7in]{geometry}
\usepackage{amsthm}
\usepackage{amsfonts}
\usepackage{graphicx}
\usepackage{epstopdf}
\usepackage[usenames,dvipsnames]{pstricks}
\usepackage{epsfig}
\usepackage{caption}
\usepackage{subcaption}
\usepackage{authblk}
\usepackage{afterpage}
\usepackage{perpage}
\usepackage{xcolor}
\usepackage{multicol}
\usepackage{listings}
\usepackage{braket}
\usepackage[switch]{lineno}
\usepackage{enumerate}
\usepackage{mathtools}
\makeatletter
\let\MYcaption\@makecaption
\makeatother

\MakePerPage{footnote}
\usepackage{amssymb}
\usepackage[font=footnotesize]{subcaption}
\newtheorem{thm}{Theorem}
\newtheorem{corollary}{Corollary}

\makeatletter
\let\@makecaption\MYcaption
\makeatother
\usepackage{url}
\usepackage{textcomp}
\usepackage{multirow}

\title{Quantum outage probability for time-varying quantum channels}
\author{Josu Etxezarreta Martinez$^{1,*}$, Patricio Fuentes$^{1}$, Pedro Crespo$^{1}$ and Javier Garcia-Frias$^{2}$}
\affil{\small{$^1$Department of Basic Sciences, Tecnun - University of Navarra, 20018 San Sebastian, Spain}}
\affil{\small{$^2$Department of Electrical and Computer Engineering, University of Delaware, Newark, DE 19716 USA}}
\affil[]{$^*$\small{\textit{jetxezarreta@tecnun.es}}}
\date{\today}

\begin{document}
\twocolumn[
  \begin{@twocolumnfalse}
\maketitle
\begin{abstract} 
Recent experimental studies have shown that the relaxation time, $T_1$, and the dephasing time, $T_2$, of superconducting qubits fluctuate considerably over time. Time-varying quantum channel (TVQC) models have been proposed in order to consider the time varying nature of the parameters that define qubit decoherence. This dynamic nature of quantum channels causes a degradation of the performance of quantum error correction codes (QECC) that is portrayed as a flattening of their error rate curves. In this article, we introduce the concepts of quantum outage probability and quantum hashing outage probability as asymptotically achievable error rates by a QECC with quantum rate $R_Q$ operating over a TVQC. We derive closed-form expressions for the family of time-varying amplitude damping channels (TVAD) and study their behaviour for different scenarios. We quantify the impact of time-variation as a function of the relative variation of $T_1$ around its mean. We conclude that the performance of QECCs is limited in many cases by the inherent fluctuations of their decoherence parameters and corroborate that parameter stability is crucial to maintain the excellent performance observed over static quantum channels.
\end{abstract}
\hfill
\end{@twocolumnfalse}
]

\section{Introduction}
The proneness of quantum information to errors puts in jeopardy the astonishing potential of quantum technologies to solve computational problems that cannot be efficiently processed by classical machines \cite{grover,shor,drugs}. Quantum errors arise due to the loss of coherence experienced by quantum states as a consequence of their interaction with the surrounding environment \cite{josuchannels}. This phenomenon is known as environmental decoherence. Quantum error correction codes were conceived as methods to protect quantum information from the deleterious effects of decoherence. Such strategies are of paramount importance to fulfil the potential of quantum technologies. In consequence, the quantum information community has gone above and beyond in its pursuit of QECCs that exhibit excellent performance and are capable of reversing quantum errors while consuming the least possible amount of resources. Several promising families of QECCs such as Quantum Reed-Muller codes \cite{QRM}, Quantum Low Density Parity Check (QLDPC) codes \cite{bicycle}, Quantum Low Density Generator Matrix (QLDGM) codes \cite{jgf,patrick,patrick2,patrick3}, Quantum Convolutional Codes (QCC) \cite{QCC}, Quantum Turbo Codes (QTC) \cite{QTC,EAQTC,josu1,josu2,josu3}, and Quantum Topological Codes \cite{toric,QEClidar} have been constructed following this premise.

Accurate mathematical modelling of decoherence effects is invaluable to construct QECCs that work in realistic scenarios. Abstractions that represent the effects of decoherence on quantum information are known as quantum channels. In the context of density matrices, quantum channels are completely-positive, trace-preserving (CPTP) linear maps between spaces of operators \cite{josuchannels}. Generally, these transformations are described via the Choi-Kraus representation as a set of matrices known as Kraus or error operators. Quantum noise models that describe decoherence effects experienced by two-level systems (qubits) in a fairly complete manner depend on the so-called relaxation time, $T_1$, and on the dephasing time, $T_2$, \cite{josuchannels}. $T_1$ and $T_2$ are experimentally measurable parameters that provide a nexus between the actual qubits that will be constructed for a quantum processor and the theoretical models that are used to describe how these qubits behave. Previous literature on QECC assumes that $T_1$ and $T_2$ are fixed parameters, implying that the quantum channels used for noise modelling are static and that their behaviour does not change over time \cite{josuchannels,QRM,bicycle,jgf,patrick,patrick2,patrick3,QCC,QTC,EAQTC,josu1,josu2,josu3,toric,QEClidar}.

Recent experimental studies on superconducting qubits have shown that $T_1$ and $T_2$ are time-variant \cite{decoherenceBenchmarking,fluctAPS,fluctApp}. These results have led to the proposal of the framework of time-varying quantum channels \cite{TVQC} as quantum channel models that fluctuate from time-to-time. This time-varying channel paradigm stands in contrast to the static approach that has been assumed previously. The TVQC model \cite{TVQC} is relevant given its consideration of the dynamic behaviour of experimentally measured decoherence parameters. Furthermore, it was shown in \cite{TVQC} that the excellent performance achieved by the QECCs proposed in the literature is compromised when the time-variations considered for decoherence modelling are significant and the QECCs have a steep error correction curve. This last result is embodied by a flattening of the steep error rate curve (as a function of the noise level) of those QECCs. In consequence, since the TVQC portrays a more realistic mathematical abstraction of the quantum noise suffered by superconducting qubits when time variations of decoherence parameters exist, the real error rate of QECCs should present this flattening effect.

In this article, we study the asymptotical limits of error correction for the paradigm of time-varying quantum channel models. Motivated by the similarity between the TVQCs and classical slow or block fading scenarios, i.e., when the channel remains constant over the duration of the coded block \cite{TVQC,tse}; we define the quantum outage probability of a TVQC as the asymptotically achievable error rate for a QECC with quantum rate $R_\mathrm{Q}$ that operates over the aforementioned noise model. Additionally, we also introduce the concept of the quantum hashing outage probability to provide an upper bound on the asymptotically achievable error rate for TVQC channels whose quantum capacity (of their static counterparts) is unknown, but for which a lower bound known as the hashing limit exists. Based on the experimentally determined statistical distribution of $T_1$ \cite{TVQC,decoherenceBenchmarking}, we provide closed-form expressions for the known TVQCs: time-varying amplitude damping channel (TVAD), time-varying amplitude damping Pauli twirl approximated channel (TVADPTA) and time-varying amplitude damping Clifford twirl approximated channel (TVADCTA). We analyze the quantum outage probability and quantum hashing outage probabilities of the aforementioned TVQCs for different scenarios. Finally, QTCs operating over the considered channels are numerically studied and benchmarked using the derived information theoretic limits.

\section{Time-varying quantum channels}\label{sec:TVQC}

The time-varying quantum channel model \cite{TVQC} has been recently proposed with the purpose of including the time fluctuations that are inherent in the decoherence parameters of superconducting qubits  \cite{decoherenceBenchmarking,fluctAPS,fluctApp}. In \cite{TVQC}, several superconducting qubit scenarios were considered depending on the influence of $T_1$ and $T_2$ on the decoherence effects experienced by these qubits. The time-varying amplitude damping (TVAD) channel was proposed for qubits whose pure dephasing rates are negligible ($T_1$-limited) and the time-varying combined amplitude and phase damping channels (TVAPD) were proposed for qubits that have pure dephasing channels that require attention ($T_1\approx T_2$ and $T_2$-dominated scenarios). In this work we will focus on the asymptotical limits for the TVAD channel.

The experimental analysis presented in \cite{decoherenceBenchmarking,TVQC} shows that $T_1(t,\omega)$ can be modeled by a WSS random process of mean $\mu_{T_1}$ and standard deviation $\sigma_{T_1}$ with a stochastic coherence time, $T_\mathrm{c}$, which is in the order of minutes. Since the processing times for quantum algorithms and error correction rounds, $t_{\mathrm{algo}}$, are in the order of microseconds \cite{TVQC}, $t_{\mathrm{algo}}\ll T_\mathrm{c}$, it is reasonable to assume that the process $T_1(t,\omega)$ remains constant during the execution of the algorithm. In other words, $T_1(\omega,t)$ can be modeled as a random variable ($t=0$ has been selected without loss of generality due to the fact that the process is WS stationary.) $T_1(\omega) = T_1(t, \omega)\rvert_{t=0},\forall t\in[0,T],T<<T_\mathrm{c}$. Given that the random process $T_1(t,\omega)$ is assumed to be Gaussian, the random variable $T_1(\omega)$ will also be Gaussian with distribution $\mathcal{N}(\mu_{T_1},\sigma^2_{T_1})$. However, since any realization of $T_1(\omega)$ should always be positive, $T_1$ must be modeled as a truncated Gaussian random variable in the region $[0,\infty]$. Therefore, the probability density function of $T_1(\omega)$ is modeled as

\begin{equation}\label{eq:pdf}
f_{T_1}(t_1)\begin{cases}
\frac{1}{\sigma_{T_1}\sqrt{2\pi}}\frac{\mathrm{e}^{-\frac{(t_1 - \mu_{T_1})^2}{2\sigma_{T_1}^2}}}{1-\mathrm{Q}\left(\frac{\mu_{T_1}}{\sigma_{T_1}}\right)} &\text{ if } t_1\geq 0 \\
0 &\text{ if } t_1<0
\end{cases},
\end{equation}
where in the above expression, $\mathrm{Q}(\cdot)$ is the Q-function defined as
\begin{equation}\label{josu6}
\mathrm{Q}(x) = \frac{1}{\sqrt{2\pi}}\int_x^{\infty}\mathrm{e}^{-\frac{x^2}{2}} dx.
\end{equation}

\subsection{Time-varying amplitude damping channels}\label{sub:TVAD}
The amplitude damping channel is a fairly complete model for describing the decoherence effects suffered by superconducting qubits \cite{josuchannels}. To be more specific, it accurately models the quantum noise experienced by qubits that are said to be $T_1$-limited. In consequence, the time-varying amplitude damping (TVAD) channel was proposed in \cite{TVQC}.
The Kraus operators for the TVAD channel for $T_1$-limited superconducting qubits are given by
\begin{equation}\label{josu7}
 E_0=\left( \begin{array}{cc}
1 & 0\\
0 & \sqrt{1-\gamma(t,\omega)} \end{array} \right),\;\;E_1=\left( \begin{array}{cc}
0 & \sqrt{\gamma(t,\omega)}\\
0 & 0 \end{array} \right),
\end{equation}
where the damping parameter WSS random process, $\gamma(t,\omega)$, is related to the relaxation WSS random process, $T_1(t,\omega)$, as

\begin{equation}\label{josu4}
\gamma(t,\omega)=1-e^{-\frac{t_{algo}}{T_1(t,\omega)}}.
\end{equation}

\subsection{Time-varying twirl approximated channels}\label{sub:TVTA}
As comprehensive as amplitude damping channels are, they cannot be efficiently implemented in classical computers when the number of qubits starts to grow. This limitation made the research community consider the use of approximated channels that are efficiently implementable in the classical domain and that maintain enough information about quantum noise. Twirling is an extensively used method in quantum information theory to study the average effect of general quantum noise
models via their mapping to more symmetric versions of themselves \cite{PTAtwirl,CTAtwirl}. Moreover, it is known that any correctable code for a twirled channel is a correctable code for the original channel up to an additional unitary correction \cite{josuchannels}.

As a consequence, Pauli twirl approximated channels (PTA) and Clifford twirl approximated channels (CTA) have been widely used in the context of quantum error correction \cite{josuchannels,PTAtwirl,CTAtwirl}. These approximated channels belong to the family of Pauli channels. Since they fulfill the Gottesman-Knill theorem \cite{NielsenChuang}, they can be simulated appropriately on classical machines \cite{josuchannels}.

Since the TVAD channel model is fairly successful in describing $T_1$-limited superconducting qubits \cite{gamma,bylander}, the time-varying amplitude damping Pauli twirl approximated channel (TVADPTA) and the time-varying amplitude damping Clifford twirl approximated channel (TVADCTA) were proposed in \cite{TVQC}. More precisely, when Pauli twirling a TVAD channel with Kraus operators in \eqref{josu7}, the resulting TVADPTA has the following Kraus operators:
\begin{equation}\label{josu8}
\{\sqrt{p_\mathrm{I}(\gamma)}\mathrm{I},\sqrt{p_\mathrm{x}(\gamma)}\mathrm{X},\sqrt{p_\mathrm{y}(\gamma)}\mathrm{Y},\sqrt{p_\mathrm{z}(\gamma)}\mathrm{Z}\},
\end{equation}
with $\sum_{k\in\{\mathrm{I,x,y,z}\}}p_k(\gamma)=1$,
and
\begin{equation}\label{josu9}
p_\mathrm{x}(\gamma)=p_\mathrm{y}(\gamma)=\frac{\gamma(t,\omega)}{4}\;\; \mbox{and $p_\mathrm{z}(\gamma)=\left(\frac{1-\sqrt{1-\gamma(t,\omega)}}{2}\right)^2$}.
\end{equation}

Notice that the TVADPTA exhibits some degree of asymmetry (Asymmetry refers to the fact that there is a mismatch between errors of type $\mathrm{Z}$ and errors of type $\mathrm{X}$ and $\mathrm{Y}$. Asymmetry is quantified by the so-called asymmetry parameter $\alpha = p_\mathrm{z}/p_\mathrm{x}$ \cite{josu3,patrick3,PTAtwirl}).

On the other hand, when Clifford twirling a TVAD channel, the resulting Kraus operators for the TVADCTA are those defined in \eqref{josu8}, with
\begin{equation}\label{josu10}
\mbox{ $p_\mathrm{I}(\gamma)=\left(\frac{1+\sqrt{1-\gamma(t,\omega)}}{2}\right)^2$, and $p_k(\gamma)=\frac{1-p_\mathrm{I}(\gamma)}{3}$},
\end{equation}
where $k\in\{\mathrm{x,y,z}\}$.
Note that the TVADCTA channels belong to the sub-family of Pauli channels known as depolarizing channels, since the additional symplectic twirl performed on the Pauli twirl in order to obtain the Clifford twirl symmetrizes the error distribution, which results in $p_\mathrm{x}=p_\mathrm{y}=p_\mathrm{z}$ \cite{josuchannels,CTAtwirl}.

We denote as $\mathbf{p}_{\mbox{\tiny{ADPTA}}}(\gamma)$ and $\mathbf{p}_{\mbox{\tiny{ADCTA}}}(\gamma)$ the probability mass functions for the ADPTA and ADCTA channels defined by \eqref{josu9} and \eqref{josu10}, respectively.

Note that since the Kraus operators of all the discussed quantum channels are a function of the relaxation time stochastic process, $T_1(t,\omega)$, they will be constant for the coherence time and are obtained by the realizations of the probability distribution in \eqref{eq:pdf}.

\section{Quantum Capacity}\label{sec:Qcap}
The quantum capacity is the maximum rate at which quantum information can be communicated/corrected over many independent uses of a noisy quantum channel. In other words, the concept of the quantum capacity establishes the quantum rate limit for which reliable (i.e., with a vanishing error rate) quantum communication/correction is asymptotically possible. The definition of the quantum capacity, often referred to as the Lloyd-Shor-Devetak (LSD) capacity, is given by the following theorem \cite{wildeQIT,quantumcap}.
\begin{thm}[LSD capacity]\label{thm:quantumcap}
\textit{The quantum capacity $C_\mathrm{Q}(\mathcal{N})$ of a quantum channel $\mathcal{N}$ is equal to the regularized coherent information of the channel
\begin{equation}\label{eq:cqreg}
C_\mathrm{Q}(\mathcal{N}) = Q_{\mathrm{reg}}(\mathcal{N}),
\end{equation}
where
\begin{equation}\label{eq:coherentreg}
Q_{\mathrm{reg}}(\mathcal{N}) = \lim_{n\rightarrow \infty} \frac{1}{n} Q_{\mathrm{coh}}(\mathcal{N}^{\otimes n}).
\end{equation}
The channel coherent information $Q_{\mathrm{coh}}(\mathcal{N})$ is defined as
\begin{equation}\label{eq:coh}
Q_{\mathrm{coh}}(\mathcal{N}) = \max_{\rho} (S(\mathcal{N}(\rho)) - S(\rho_\mathrm{E})),
\end{equation}
where $S$ is the von Neumann entropy and $S(\rho_\mathrm{E})$ measures how much information the environment has.}
\end{thm}

For general channels, there is no closed-form analytical expression of the quantum capacity given in theorem \ref{thm:quantumcap}. However, the AD channel and its twirl approximations have either closed-form expressions or bounds for their LSD capacities.

\subsection{Static amplitude damping channel}\label{sub:CqAD}
The quantum capacity of an AD channel with damping parameter $\gamma \in [0,1]$ is equal to \cite{wildeQIT,quantumcap}
\begin{equation}\label{eq:ADcap}
C_\mathrm{Q}(\gamma) = \max_{\xi\in[0,1]}  \mathrm{H}_2((1-\gamma) \xi) - \mathrm{H}_2(\gamma \xi),
\end{equation}
whenever $\gamma \in [0, 1/2]$, and zero for $\gamma\in [1/2,1]$. $\mathrm{H}_2(x)$ is the binary entropy.

\subsection{Static Pauli channels}\label{sub:CqP}
An expression for the quantum capacity of the widely used Pauli channels remains unknown \cite{josuchannels,wildeQIT}. However, a lower bound that can be achieved by stabilizer codes, the hashing bound, $C_\mathrm{H}$, \cite{wildeQIT} is known. The reason why the quantum capacity of a Pauli channel can be higher than the hashing bound, i.e. $C_\mathrm{Q} \geq C_\mathrm{H}$, is the degenerate nature of quantum codes \cite{degenPRL,degen}, which arises from the fact that several distinct channel errors affect quantum states in an indistinguishable manner.

The hashing bound for a Pauli channel defined by the probability mass function $\mathbf{p} = (p_\mathrm{I},p_\mathrm{x},p_\mathrm{y},p_\mathrm{z})$ is given by \cite{wildeQIT}
\begin{equation}\label{eq:hash}
C_\mathrm{H}(\mathbf{p}) = 1 - \mathrm{H}_2(\mathbf{p}).
\end{equation}
$\mathrm{H}_2(\mathbf{p}) = -\sum_j p_j\log_2(p_j)$ is the entropy in bits of a discrete random variable with probability mass function given by $\mathbf{p}$.

Equation \eqref{eq:hash} speaks towards the general hashing bound of the whole family of Pauli channels. Sub-families of Pauli channels of special interest are the ones obtained by $\mathcal{P}_n$- and $\mathcal{C}_1^{\otimes n}$-twirling the AD channel \cite{josuchannels,TVQC,PTAtwirl,CTAtwirl}. The families of Pauli channels obtained by such operations are denominated the AD Pauli twirl approximated (ADPTA) channel and the AD Clifford twirl approximated (ADCTA) channel, which is a depolarizing channel (since $p_x=p_y=p_z$). The parameters $p_x,p_y,p_z$ are themselves functions of the relaxation time, $T_1$, due to the fact that they are approximated by the $T_1$-dependent AD channel.

\section{Quantum outage probability}\label{sub:quantumpout}
The TVQC model proposed in \cite{TVQC} clearly resembles the paradigm of classical block fading \cite{tse}. This occurs because the stochastic processes that define the dynamic behaviour of the relaxation and dephasing times are considered to be constant for the codeword length, since the processing times are much smaller than the coherence times of such processes \cite{TVQC,decoherenceBenchmarking}. Also, it is considered that all qubits in the codeword are affected equally by the noise \cite{TVQC}. This is reminiscent of classical block fading scenarios in which the channel gain, $h$, is considered to be constant for the codeword length \cite{tse}. We use this connection with the classical domain to develop the information theoretical concepts for TVQCs.

\subsection{Classical outage probability for the block fading additive Gaussian noise channel}\label{subsub:classicalpout}

Under slow fading conditions, the channel gain of an Aditive White Gaussian Noise (AWGN) channel, which is generally modeled as a wide-sense stationary  (WSS) random process $\alpha(t,\omega)$, varies slowly with respect to the time duration of a codeword. In these situations, the value of the channel gain during the transmission of a codeword can be considered to be approximately constant and given by a realization of the random variable $\alpha(\omega)$. Therefore, the block fading channel can be reduced to an AWGN channel where the received signal-to-noise ratio is a random variable $|\alpha(\omega)|^2 \mathrm{SNR}$. Consequently, the channel capacity also becomes the random variable $C(\omega) = \log_2(1 + |\alpha(\omega)|^2 \mathrm{SNR})$ with bits per channel use serving as the measuring units.
Note that by the Shannon channel coding theorem, given an encoding rate $R$ bits per channel use, reliable communication will be possible if the realization of the channel capacity $C(\omega)$ is larger than $R$. On the other hand, when $C<R$ communication with low probability of error is not possible. The probability that communications fail when transmitting a codeword with rate $R$ is called outage probability and is given by \cite{tse}
\begin{equation}\label{eq:poutFading}
\begin{split}
p_{\mathrm{out}}(R,\mathrm{SNR}) &= \mathrm{P}(\{\omega\in\Omega : C(\omega) < R\}) \\  &= P(\{\omega\in\Omega : \log_2(1+|\alpha(\omega)|^2 \mathrm{SNR}) < R\}).
\end{split}
\end{equation}

The outage probability will depend on the probability distribution of the channel gain random variable, $\alpha(\omega)$. For the widely used Rayleigh fading model, for which the channel gain follows a circularly symmetric complex normal distribution, $\mathcal{CN}(0,1)$, the outage probability can be shown to be equal to \cite{tse}
\begin{equation}\label{eq:rayleighPout}
p_{\mathrm{out}}(R,\mathrm{SNR}) = 1 - \mathrm{e}^{\frac{-(2^R - 1)}{\mathrm{SNR}}}.
\end{equation}

\subsection{Quantum outage probability}\label{sub:poutQdef}
Based on the similarity with the classical block fading scenario, we can assume that each of the realizations of the qubit relaxation and dephasing times, $T_1$ and $T_2$, will result in a realization of the time-varying quantum channel in consideration, and consequently, in a specific value for the quantum channel capacity, $C_\mathrm{Q}$ qubits per channel use. Similar to classical coding, if the realization of the decoherence parameters leads to a channel capacity lower that the quantum coding rate, $R_\mathrm{Q}$ qubits per channel use, then the quantum bit error rate (QBER) will not vanish asymptotically with the blocklength, independently of the selected QECC. Thus, we can state that for such realizations the channel will be in outage. Therefore, we define
\begin{equation}\label{eq:qPout}
p_{\mathrm{out}}^\mathrm{Q}(R_\mathrm{Q}) = \mathrm{P}(\{\omega\in\Omega : C_\mathrm{Q}(\omega) < R_\mathrm{Q}\}),
\end{equation}
as the quantum outage probability.

In other words, with probability $p_{\mathrm{out}}^\mathrm{Q}(R_\mathrm{Q})$, the capacity of the channel $C_\mathrm{Q}(\omega)$ will be lower than the rate of the code, and thus, the error rate will not vanish asymptotically. Conversely, with probability $1-p_{\mathrm{out}}^\mathrm{Q}(R_\mathrm{Q})$, reliable quantum correction will be possible. Thus, the quantum outage probability will be the asymptotically achievable error rate for quantum error correction when the rate is $R_\mathrm{Q}$.

\section{Computation of the quantum outage probability for the family of time-varying amplitude damping channels}
Next, we derive the quantum outage probability for the family of TVAD channels in Theorem \ref{thm:poutTVAD} \cite{TVQC} and we provide a closed-form expression for this quantity when the TVAD is considered. In addition, we define the quantum hashing outage probability as a bound of $p_{\mathrm{out}}^\mathrm{Q}$ for the twirl aproximated TVAD Pauli channels, as their exact LSD capacities are not known.

\subsection{Outage probability for the time-varying amplitude damping channel}\label{sec:poutTVAD}

It is important to define a set of specific concepts before Theorem \ref{thm:poutTVAD} is introduced. It is clear from expression \eqref{eq:ADcap}, that the quantum capacity $C_\mathrm{Q}$ of the AD channel is a monotonically decreasing function of the damping parameter, $\gamma$. Therefore, there will be a unique $\gamma^*(R_\mathrm{Q})$ that makes the value of the channel capacity $C_\mathrm{Q}$ equal to $R_\mathrm{Q}$, i.e., $C_\mathrm{Q}(\gamma^*(R_\mathrm{Q})) = R_\mathrm{Q}$. That is to say,
\begin{equation}\label{josu2}
C_\mathrm{Q}(\gamma^*(R_\mathrm{Q})) = R_\mathrm{Q} \Leftrightarrow\gamma^*(R_\mathrm{Q})=C^{-1}_\mathrm{Q}(R_\mathrm{Q}).
\end{equation}
We will refer to $\gamma^*(R_\mathrm{Q})$ as the noise limit. Notice that codes with rates $R_\mathrm{Q}$ cannot operate reliably for channels noisier than the noise limit, where by noisier we mean that the channel has a higher value of the damping parameter, $\gamma$ (note that $\gamma$ describes how intense the amplitude damping effects are).

Additionally, from \eqref{josu4} we define the critical relaxation time $T_1^*(R_\mathrm{Q}, t_{\mathrm{algo}})$ as
\begin{equation}\label{eq:citicalT1}
T_1^*(R_\mathrm{Q},t_{\mathrm{algo}}) = \frac{-t_{\mathrm{algo}}}{\ln{(1-\gamma^*(R_\mathrm{Q}))}},
\end{equation}
which is a function of the algorithm time $t_{\mathrm{algo}}$. In order to perform accurate comparisons of quantum channels with different mean relaxation times, $\mu_{T_1}$, we rewrite the critical time as a function of the damping parameter, $\gamma$, that the AD channel exhibits when its static version is considered (This is similar to the normalization done in \cite{TVQC}, as the damping rate is a function of the algorithm time and the relaxation time). Note that if the calculations were done as a function of the algorithm time, the comparison between qubits with different mean relaxation times would not be ideal since for a fixed $t_{algo}$, higher values of $\mu_{T_1}$ result in lower values of $\gamma $ and thus, lower channel noise. It is obvious that longer mean relaxation times are more favorable for computing applications, as they allow for longer algorithm times. However, we are interested in calculating the quantum outage probability versus the noise level of the channel, i.e, we want to know how much noise a qubit is able to tolerate. Consequently, we obtain the time that the quantum algorithm would require to reach a noise level $\gamma$ for the static AD channel as \cite{josuchannels}
\begin{equation}\label{eq:talgo}
t_{\mathrm{algo}} = - \mu_{T_1}\ln{(1 - \gamma)}.
\end{equation}

This way, the critical relaxation time in \eqref{eq:citicalT1} is a function of the damping parameter $\gamma$ associated to the static AD channel as
\begin{equation}\label{josu12}
T^*_{1}(R_\mathrm{Q},\gamma)=\frac{\mu_{T_1}\ln(1-{\gamma})}{\ln(1-{\gamma^*(R_\mathrm{Q})})}.
\end{equation}

Finally, the coefficient of variation, $c_\mathrm{v}$, was shown in \cite{TVQC} to be the most relevant parameter to describe how much the variations of $T_1$ influence the error correcting performance of QECCs. The coefficient of variation of a random variable is a standardized measure of dispersion of a probability distribution and it is defined as
\begin{equation}\label{eq:CVar}
c_\mathrm{v} = \frac{\sigma}{\mu},
\end{equation}
where $\sigma$ is the standard deviation of the random variable and $\mu$ is its mean. This parameter measures the extent to which realizations of a random variable can deviate from its mean.

At this point, we are ready to introduce the theorem that provides the quantum outage probability for TVAD channels that consider the qubits of \cite{decoherenceBenchmarking,TVQC}.
\begin{thm}[TVAD quantum outage probability]\label{thm:poutTVAD}
\textit{The quantum outage probability for the time-varying amplitude damping channels associated to the damping parameter,  $\gamma\in[0,1-{e}^{-1}]$, is equal to
\begin{equation}\label{eq:poutTVAD}
p_{\mathrm{out}}^\mathrm{Q}(R_\mathrm{Q}, \gamma) = 1 - \frac{\mathrm{Q}\left(\frac{1}{c_\mathrm{v}(T_1)}\left[\frac{\ln(1-{\gamma})}{\ln(1-{\gamma^*(R_\mathrm{Q})})} - 1\right]\right)}{1 - \mathrm{Q}\left(\frac{1}{c_\mathrm{v}(T_1)}\right)},
\end{equation}
where $\mathrm{Q}(\cdot)$ is the Q-function, $c_\mathrm{v}(T_1)$ is the coefficient of variation of $T_1$ \eqref{eq:CVar}, $\gamma^*(R_\mathrm{Q})$ is the noise limit, $\mu_{T_1}$ is the mean relaxation time and $\sigma_{T_1}$ is the standard deviation of the relaxation time.}
\end{thm}
\begin{proof}
In order to compute the outage probability $p_{\mathrm{out}}^\mathrm{Q}(R_\mathrm{Q},\gamma)$, we use the fact of the decreasing monotonicity of $C_\mathrm{Q}$ and $T_1$ with respect to $\gamma$. This implies that the events $\{\omega\in\Omega : C_\mathrm{Q}(\gamma(\omega))< R_\mathrm{Q}\}$, $\{ \omega\in\Omega : \gamma(\omega)< \gamma^*(R_\mathrm{Q})\}$ and $\{\omega\in\Omega : T_1(\omega) < T_1^*(R_\mathrm{Q},\gamma)\}$ are all the same.
Therefore,
\begin{equation}\label{eq:TVADstep1}
\begin{split}
p_{\mathrm{out}}^\mathrm{Q}(R_\mathrm{Q}, \gamma) &= \mathrm{P}(\{\omega\in\Omega : C_\mathrm{Q}(\omega) < R_\mathrm{Q}\}) \\&=  \mathrm{P}(\{\omega\in\Omega : T_1(\omega) < T_1^*(R_\mathrm{Q},\gamma)\}).
\end{split}
\end{equation}

Next, we compute \eqref{eq:TVADstep1} based on the fact that the random variable $T_1(\omega)$ is modelled by the probability density function in equation \eqref{eq:pdf} \cite{TVQC}.


The outage probability of the TVAD channel can be calculated as

\begin{equation}\label{eq:TVADstep3}
\begin{split}
p_{\mathrm{out}}^\mathrm{Q}(&R_\mathrm{Q},\gamma) =  \mathrm{P}(\{\omega\in\Omega : T_1(\omega) < T_1^*(R_\mathrm{Q}, \gamma)\}) \\
& = \int_{-\infty}^{T_1^*(R_\mathrm{Q}, \gamma)}f_{T_1}(t_1)dt_1 \\
 &= \int_{0}^{T_1^*(R_\mathrm{Q},\gamma)}\frac{1}{\sigma_{T_1}\sqrt{2\pi}}\frac{\mathrm{e}^{-\frac{(t_1 - \mu_{T_1})^2}{2\sigma_{T_1}^2}}}{1-\mathrm{Q}\left(\frac{\mu_{T_1}}{\sigma_{T_1}}\right)} dt_1 \\
&=\frac{1}{1-\mathrm{Q}\left(\frac{\mu_{T_1}}{\sigma_{T_1}}\right)}\left(\int_{-\infty}^{T_1^*(R_\mathrm{Q}, \gamma)}\frac{1}{\sigma_{T_1}\sqrt{2\pi}}\mathrm{e}^{-\frac{(t_1 - \mu_{T_1})^2}{2\sigma_{T_1}^2}}dt_1\right. \\ & \left.- \int_{-\infty}^{0}\frac{1}{\sigma_{T_1}\sqrt{2\pi}}\mathrm{e}^{-\frac{(t_1 - \mu_{T_1})^2}{2\sigma_{T_1}^2}}dt_1\right)\\
& =\frac{1}{1-\mathrm{Q}\left(\frac{\mu_{T_1}}{\sigma_{T_1}}\right)}\left(\int_{-\infty}^{\frac{T_1^*(R_\mathrm{Q}, \gamma)-\mu_{T_1}}{\sigma_{T_1}}}\frac{1}{\sqrt{2\pi}}\mathrm{e}^{-\frac{\eta^2}{2}}d\eta\right. \\ & \left. - \int_{-\infty}^{\frac{-\mu_{T_1}}{\sigma_{T_1}}}\frac{1}{\sqrt{2\pi}}\mathrm{e}^{-\frac{\eta^2}{2}}d\eta\right)\\
& = \frac{1}{1-\mathrm{Q}\left(\frac{\mu_{T_1}}{\sigma_{T_1}}\right)}\left(1-\int_{\frac{T_1^*(R_\mathrm{Q}, \gamma)-\mu_{T_1}}{\sigma_{T_1}}}^{\infty}\frac{1}{\sqrt{2\pi}}\mathrm{e}^{-\frac{\eta^2}{2}}d\eta\right. \\ & \left. - \int_{\frac{\mu_{T_1}}{\sigma_{T_1}}}^{\infty}\frac{1}{\sqrt{2\pi}}\mathrm{e}^{-\frac{\eta^2}{2}}d\eta\right) \\ 
& =\frac{1-\mathrm{Q}\left(\frac{T_1^*(R_\mathrm{Q}, \gamma)-\mu_{T_1}}{\sigma_{T_1}}\right) - \mathrm{Q}\left(\frac{\mu_{T_1}}{\sigma_{T_1}}\right)}{1-\mathrm{Q}\left(\frac{\mu_{T_1}}{\sigma_{T_1}}\right)} \\ & = 1 - \frac{\mathrm{Q}\left(\frac{T_1^*(R_\mathrm{Q}, \gamma)-\mu_{T_1}}{\sigma_{T_1}}\right)}{1-\mathrm{Q}\left(\frac{\mu_{T_1}}{\sigma_{T_1}}\right)}\\&= 1 - \frac{\mathrm{Q}\left(\frac{\mu_{T_1}}{\sigma_{T_1}}\left[\frac{\ln(1-{\gamma})}{\ln(1-{\gamma^*(R_\mathrm{Q})})} - 1\right]\right)}{1 - \mathrm{Q}\left(\frac{\mu_{T_1}}{\sigma_{T_1}}\right)} \\ & =1 - \frac{\mathrm{Q}\left(\frac{1}{c_\mathrm{v}(T_1)}\left[\frac{\ln(1-{\gamma})}{\ln(1-{\gamma^*(R_\mathrm{Q})})} - 1\right]\right)}{1 - \mathrm{Q}\left(\frac{1}{c_\mathrm{v}(T_1)}\right)},
\end{split}
\end{equation}
as we wanted to prove.

It can be seen that the quantum outage probability as a function of the damping parameter $\gamma$ does not depend on the absolute value of the mean relaxation time, but on the coefficient of variation of $T_1$. This way, we decouple the time-varying effects from the fact that longer mean relaxation times admit longer quantum algorithm processing times. Consequently, we present a result agnostic to the impact that longer coherence times have and we can provide conclusions for all superconducting qubits.

To finish, one needs to make sure that under the normalization done, the maximum value that $t_{\mathrm{algo}}$ can take is still much lower than the coherence time $T_\mathrm{c}$ of the random process $T_1(t,\omega)$, which is of order of minutes \cite{decoherenceBenchmarking}. Considering algorithm times longer than the mean relaxation time makes no sense since for such timeframe the qubit is in equilibrium state with high probability, and, therefore, it is useless as a resource. Additionally, the value of the mean relaxation time is in the order of microseconds for superconducting qubits \cite{TVQC}, and so taking $t_{\mathrm{algo}}^{\mathrm{max}}=\mu_{T_1}$ for our theorem makes sense. Such algorithm time is associated to the value $\gamma = 1- e^{-1}$, so the quantum outage probability defined here is valid for the range $\gamma\in[0,1-e^{-1}]$.
\end{proof}

\subsection{Quantum hashing outage probability for the time-varying twirl approximated channels}\label{sub:TVtwirls}
Because the analytical expression for the LSD capacity of Pauli channels is not known, the quantum outage probability for this family of approximated channels cannot be calculated. Nevertheless, by means of the hashing bound, we define the quantum hashing outage probability for Pauli channels as
\begin{equation}\label{eq:poutPauli}
p_{\mathrm{out}}^\mathrm{H}(R_\mathrm{Q}) = \mathrm{P}(\{\omega\in\Omega : C_\mathrm{H}(\omega)< R_\mathrm{Q}\}).
\end{equation}

Note that the quantum hashing outage probability will be an upper bound on the actual quantum outage probability of time-varying Pauli channels, since the hashing limit is a lower bound of the LSD capacity. This way, events that exceed the hashing bound will be more likely than the events that exceed the LSD capacity, which means that $p_{\mathrm{out}}^\mathrm{H}(R_\mathrm{Q}) \geq p_{\mathrm{out}}^\mathrm{Q}(R_\mathrm{Q})$. Consequently, $p_{\mathrm{out}}^\mathrm{H}(R_\mathrm{Q})$ is an upper bound of interest for benchmarking the behaviour of the TV Pauli channels.

It is important to realise that the Hashing bound \eqref{eq:hash}
\begin{equation}
C_\mathrm{H}(\mathbf{p}(\gamma)) =1+\sum_{k\in\{\mathrm{I,x,y,z}\}}p_k(\gamma)\log_2p_k(\gamma) = 1- \mathrm{H}_2(\mathbf{p}(\gamma))
\end{equation}
 for the twirled approximated channels of the AD channel with the probability distributions given in \eqref{josu9} and \eqref{josu10}, is a monotonic decreasing function of the damping probability $\gamma$. This is justified by the fact that, as $\gamma\in[0,1]$ increases, the values of $p_\mathrm{x},p_\mathrm{y},p_\mathrm{z}$ in either \eqref{josu9} or \eqref{josu10} also increase. This results in the uncertainty of the discrete random variables associated to each of these distributions, and consequently their corresponding entropy values, becoming higher. Therefore, as for the AD channel, we define the noise limit for these Pauli channels as the unique value of the damping parameter, $\gamma_{\mathrm{T}}^*(R_\mathrm{Q})$, such that:
\begin{equation}\label{josu11}
1-C_\mathrm{H}(\mathbf{p}(\gamma^*_{\mathrm{T}}(R_\mathrm{Q})) =R_\mathrm{Q}\Leftrightarrow \gamma^*_{\mathrm{T}}(R_\mathrm{Q})=C^{-1}_\mathrm{H}(1-R_\mathrm{Q}).
\end{equation}
From \eqref{josu4}, the critical relaxation time (note that we have added the subindex $\mathrm{T}$ to indicate we are twirling the AD channel) is
\begin{equation}\label{josu122}
T^*_{1,\mathrm{T}}(R_\mathrm{Q},t_{\mathrm{algo}})=\frac{-t_{\mathrm{algo}}}{\ln(1-\gamma^*_{\mathrm{T}}(R_\mathrm{Q}))},
\end{equation}
 where the probability mass function $\mathbf{p}$ in \eqref{josu11} should be taken as $\mathbf{p}_{\mbox{\tiny{ADPTA}}}$ or $\mathbf{p}_{\mbox{\tiny{ADCTA}}}$ when considering the twirled ADPTA or ADCTA channels, respectively. Similarly to the TVAD channel, we can write the critical relaxation time as a function of the damping parameter
 \begin{equation}\label{josu12T}
T^*_{1,\mathrm{T}}(R_\mathrm{Q},\gamma)=\frac{\mu_{T_1}\ln(1-{\gamma})}{\ln(1-{\gamma^*_{\mathrm{T}}(R_\mathrm{Q})})}.
\end{equation}

The following corollary yields the hashing outage probability of the twirled TVADPTA and TVADCTA Pauli channels for the qubits studied in \cite{decoherenceBenchmarking,TVQC}.

\begin{corollary}[Quantum hashing outage probability]\label{cor:pouHPauli}
\textit{The quantum hashing outage probability for the time-varying twirled approximated channels associated to the damping parameter,  $\gamma\in[0,1-e^{-1}]$, is equal to
\begin{equation}\label{eq:poutTVADTAs}
p_{\mathrm{out}}^\mathrm{H}(R_\mathrm{Q}, \gamma) = 1 - \frac{\mathrm{Q}\left(\frac{1}{c_\mathrm{v}(T_1)}\left[\frac{\ln(1-{\gamma})}{\ln(1-{\gamma^*_{\mathrm{T}}(R_\mathrm{Q})})} - 1\right]\right)}{1 - \mathrm{Q}\left(\frac{1}{c_\mathrm{v}(T_1)}\right)},
\end{equation}
where $\mathrm{Q}(\cdot)$ is the Q-function, $c_\mathrm{v}(T_1)$ is the coefficient of variation of $T_1$ given in \eqref{eq:CVar}, $\gamma^*_{\mathrm{T}}(R_\mathrm{Q})$ is the noise limit that depends on the considered twirled approximation, $\mu_{T_1}$ is the mean relaxation time, and $\sigma_{T_1}$ is the standard deviation of the relaxation time.}
\end{corollary}
\begin{proof}
In order to compute the hashing outage probability $p_{\mathrm{out}}^\mathrm{H}(R_\mathrm{Q}, \gamma)$, we use the fact of the decreasing monotonicity of $C_\mathrm{H}$ and $T_1$ with respect to $\gamma$. This implies that the events $\{\omega\in\Omega : C_\mathrm{H}(\omega)< R_\mathrm{Q}\}$, $\{ \omega\in\Omega : \gamma(\omega)< \gamma^*_{\mathrm{T}}(R_\mathrm{Q})\}$ and $\{\omega\in\Omega : T_1(\omega) < T_{1,\mathrm{T}}^*(R_\mathrm{Q}, \gamma)\}$ are all same.
Therefore,
\begin{equation}\label{eq:TVADTAstep2}
\begin{split}
p_{\mathrm{out}}^\mathrm{H}(R_\mathrm{Q}, \gamma) &= P(\{\omega\in\Omega : C_\mathrm{H}(\omega) < R_\mathrm{Q}\})= \\&  \mathrm{P}(\{\omega\in\Omega : T_1(\omega) < T_{1,\mathrm{T}}^*(R_\mathrm{Q}, \gamma)\}
\end{split}
\end{equation}
Thus, the hashing outage corresponds to events where the realization of the relaxation time is lower than the critical relaxation time.

From this point, the calculation of the quantum hashing outage probability is the same as in the proof of Theorem \ref{thm:poutTVAD}, since equation \eqref{eq:TVADTAstep2} is the same as \eqref{eq:TVADstep1}. 
\end{proof}

Note that even though the final expression is the same as the one of Theorem \ref{thm:poutTVAD}, the noise limit value is calculated in a different manner, which means that the results are different for each of the TV channels.

\subsection{Numerical simulations}\label{sub:numsim}

By using the results derived previously, we now discuss the behaviour of the quantum outage probability and the quantum hashing outage probability. Following the reasoning of \cite{TVQC}, we will compare scenarios described by different coefficients of variation of the random variable $T_1$. For the analysis conducted in this section, we consider the following values of the coefficient of variation; $c_\mathrm{v}(T_1)=\{1,10,15,20,25\}\%$.

\subsubsection{Quantum outage probability of the TVAD channel}\label{subsub:tvadPoutQ}

Figure \ref{fig:ADpoutQ} plots the quantum outage probability versus the damping parameter, $10^{-3}\leq \gamma\leq 0.6$, for a quantum rate $R_\mathrm{Q}=1/9$, and for all the coefficients of variation of the relaxation time random variable.

\begin{figure}[!ht]
\centering
\includegraphics[width=\linewidth]{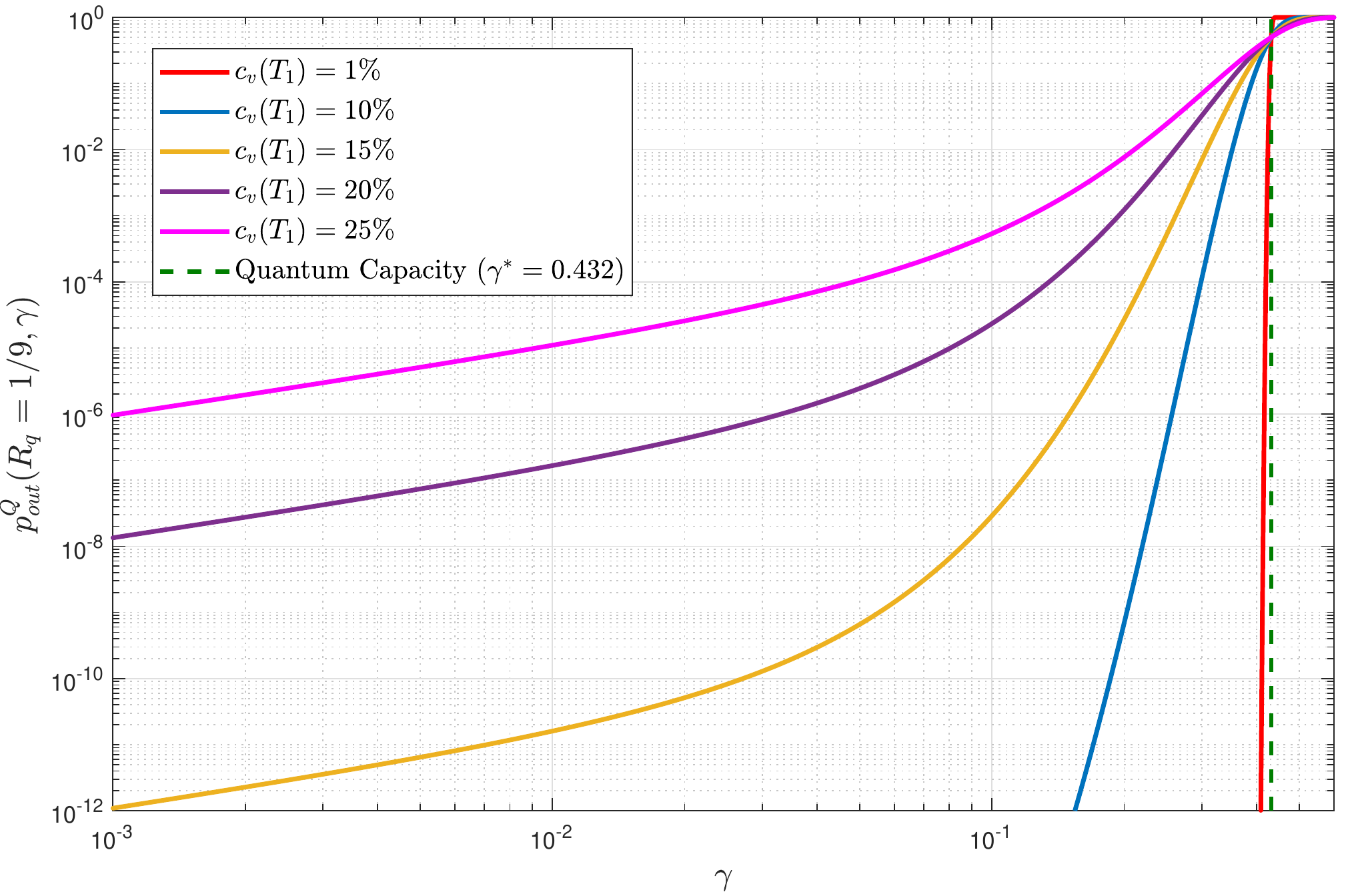}
\caption{\textbf{Quantum outage probability of the TVAD channel.} The metric is calculated for TVADs with $c_\mathrm{v}(T_1)=\{1,10,15,20,25\}\%$ and for a quantum rate of $R_\mathrm{Q}=1/9$.}
\label{fig:ADpoutQ}
\end{figure}

Figure \ref{fig:ADpoutQ} further cements the conclusions derived in \cite{TVQC}, as it shows that the impact of the fluctuations of the decoherence parameters can be accurately quantified by the coefficient of variation of $T_1$ (note that in \cite{TVQC} this was true if the waterfall region was steep enough). Notice that when the coefficient of variation is very low, i.e. $c_\mathrm{v}(T_1)=1\%$, the quantum outage probability of the TVAD channel almost coincides with the quantum capacity (represented herein by the noise limit, $\gamma^*$). Consequently, QECCs operating over TVQCs that present low coefficients of variation will behave asymptotically in a similar manner to static channels. Nevertheless, increasing the variability of the relaxation time around the mean causes the outage probability to diverge from the static capacity. In this case, the asymptotical bounds flatten and the achievable error rate of QECCs operating over TVQCs does not vanish. Therefore, the higher $c_\mathrm{v}(T_1)$ is, the worse the achievable error rate will be.

The previous discussion indicates that the coefficient of variability of the random variable $T_1(\omega)$ can be used to describe the effect that decoherence parameter time fluctuations will produce on the aymptotical limits of QECCs. These results also speak towards the importance of qubit construction and cooldown: if optimized correctly, the fluctuations relative to the mean will be mild, and the outage scenarios will be significantly less frequent. Naturally, it is desirable for qubits to exhibit long mean coherence times $T_1$ so that algorithms with longer lifespans can be handled appropriately. However, aside from seeking to increase the coherence time of qubits, it is clear that minimizing the dispersion of this parameter will be critical if these qubits are to be reliable \cite{TVQC}.

\begin{figure}[!ht]
\centering
\includegraphics[width=\linewidth]{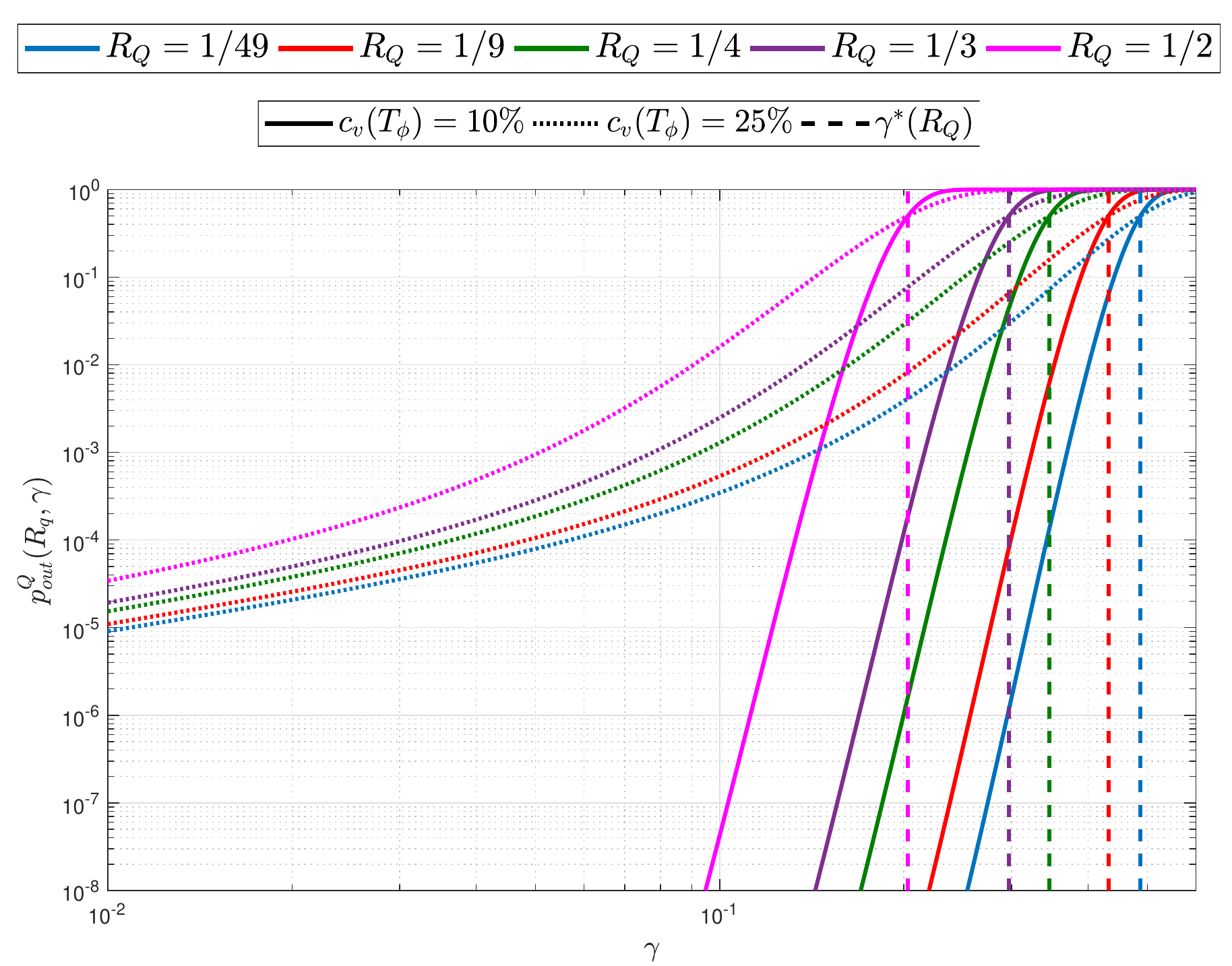}
\caption{\textbf{Quantum outage probability of the TVAD channel for different quantum rates.} The considered quantum rates are $R_\mathrm{Q}\in\{1/49,1/9,1/4,1/3,1/2\}$. We plot the quantum outage probability for $c_\mathrm{v}(T_1) =\{10,25\}\%$. The quantum capacities (noise limits, $\gamma^*$) for the static quantum channels are also represented.}
\label{fig:ADpoutRates}
\end{figure}

Let us now discuss how the quantum rate affects the quantum outage probability. Once more, Figure \ref{fig:ADpoutRates} shows the quantum outage probability for $c_\mathrm{v}(T_1) = \{10,25\}\%$, but it now considers different quantum rates $R_\mathrm{Q}\in\{1/49,1/9,1/4,1/3,1/2\}$. As expected, the results portrayed in this figure show that increasing the rate leads to an increase of $p_{\mathrm{out}}^\mathrm{Q}(R_\mathrm{Q},\gamma)$, although the shape of the curves remains similar to scenarios with the same coefficient of variation. The main takeaway is that, although increasing the rate of an error correction code reduces the overall resource consumption, this occurs at the expense of a degradation in the asymptotical error correction performance. It is important to mention that this degradation does not occur because there is higher sensitivity to time fluctuations at higher rates. Further inspection of figure \ref{fig:ADpoutRates} reveals that the noise limits for each rate change as expected, and that the outage probabilities behave similarly according to those noise limits. Thus, similarly to classical coding, the quantum rate does indeed impact the quantum outage probability, but not due to a higher sensitvity to time-variance. Furthermore, similar to what happens in static channels, there is trade-off between resource consumption and how demanding (in terms of noise) the quantum channel is.

\subsubsection{Quantum hashing outage probability of the time-varying twirl approximated channels}\label{subsub:PauliPoutH}
We continue by comparing the outage of the TVAD channel and its twirled approximated channels. Figure \ref{fig:comparison} plots the hashing outage probability results of the TVADPTA and TVADCTA channels. Note that the x-axis is still $\gamma$, despite the fact that the defining parameter for the TVADPTA and TVADCTA channels is $\mathbf{p}$. However, the $\gamma$ associated to a given $p$ can be obtained easily \cite{josuchannels}, which is necessary to perform comparisons with the TVAD channel. The quantum outage capacities of the TVAD channel from figure \ref{fig:ADpoutQ} are also shown. Note that the hashing outage probabilities for the ADPTA and ADCTA channels are worse than the quantum outage probability of the TVAD channel, since the noise limits for those channels are lower than the one for the TVAD. This means that the hashing outage probabilities of the twirled channels are worse because their noise limits are worse, and not because they are more sensitive to time-fluctuations. This is analogous to the previous explanation on the difference between the values of $p_{\mathrm{out}}^\mathrm{Q}$ for QECCs with different quantum rates operating over the TVAD channel. Also note that even though the hashing outage probabilities for these approximated channels are higher than the quantum outage probability for the TVAD, one can not conclude that the actual quantum outage probability for these twirled channels will be worse than the one for the AD channel (recall that the hashing outage probability provides an upper bound on the actual outage probability).

\begin{figure}[!h]
\centering
\includegraphics[width=\linewidth]{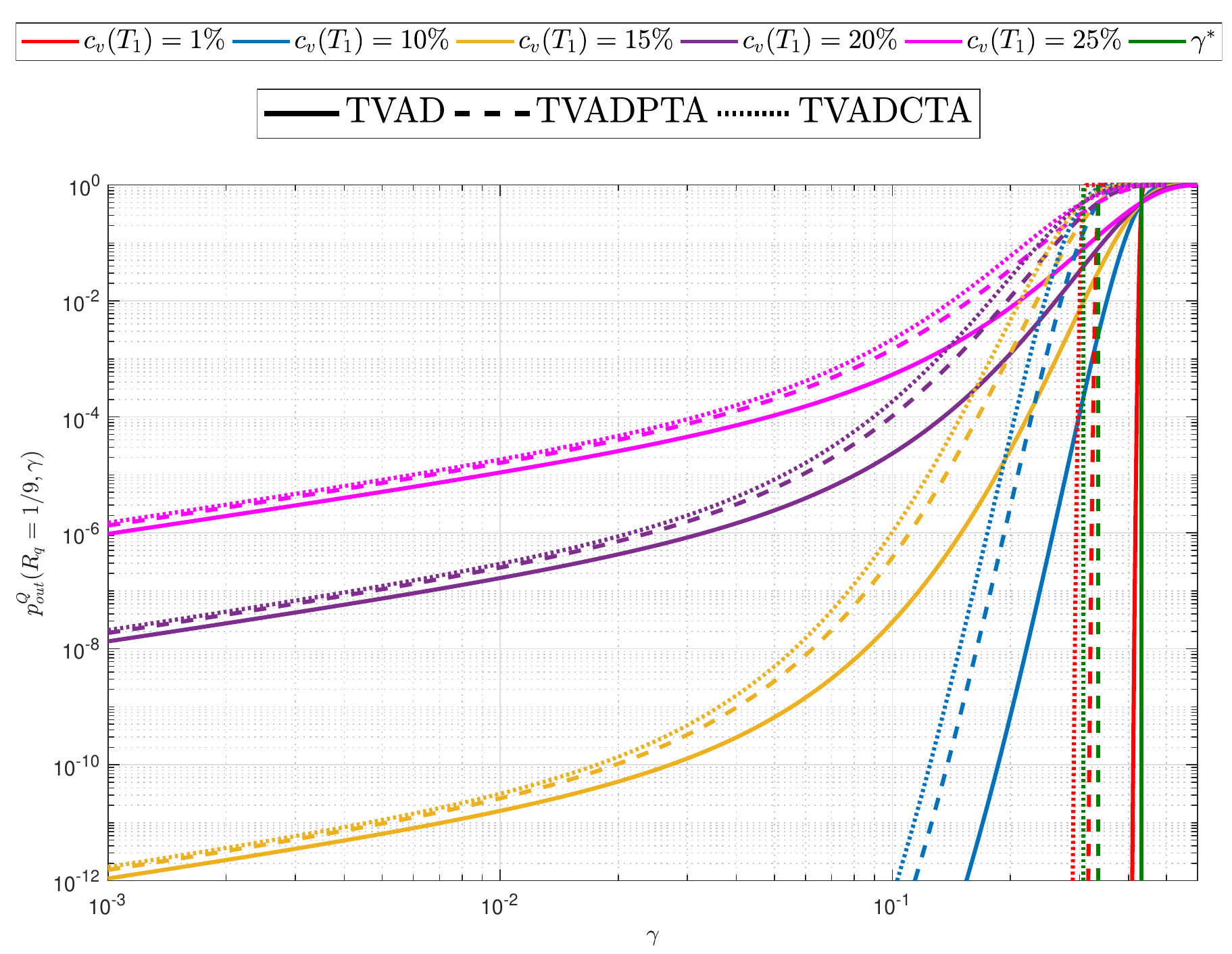}
\caption{\textbf{Quantum outage and hashing outage probabilities for TVAD, TVADPTA and TVADCTA when $\mathbf{R_\text{\textbf{Q}}=1/9}$.} The noise limits are: $\gamma^*_{\mathrm{AD}} = 0.432$, $\gamma^*_{\mathrm{ADPTA}} = 0.3354$ and $\gamma^*_{\mathrm{ADCTA}}=0.3065$.}
\label{fig:comparison}
\end{figure}

\section{QECCs operating over TVQCs and quantum outage probability}\label{sec:QECCs}
To finish our discussion on the quantum outage probability, we study the performance of a quantum turbo code of rate $R_Q=1/9$ \cite{josu3} when operating under TVQC models of decoherence \cite{TVQC} and we use the results obtained in the previous section for the quantum outage probability to benchmark its performance. Since the depolarizing channel model is the most popular error model when it comes to studying the performance of QECC families in the literature \cite{josuchannels}, we will also follow this trend herein. We consider the ADCTA depolarizing channel as the static channel of interest and the TVADCTA as its time-varying version. Considering that these decoherence models belong to the family of Pauli channels, the information theoretical benchmarks that will be considered are the hashing bound and the quantum hashing outage probability. The time-varying channels that we consider are the ones associated to the QA\_C5 ($c_v(T_1)\approx 26\%$) and QA\_C6 ($c_v(T_1)\approx 22\%$) superconducting qubit scenarios of \cite{decoherenceBenchmarking}. We select these scenarios in order to portray the performance that error correction codes would exhibit when operating on real hardware that exhibits time fluctuations \cite{decoherenceBenchmarking,TVQC}.

Monte Carlo computer simulations have been carried out to estimate the performance for the different scenarios presented in the paper. Each round (i.e., transmitted block) of the numerical simulation is performed by generating an $N$-qubit Pauli operator, calculating its associated syndrome, and finally running the decoding algorithm. Once the logical error is estimated, it is compared with the logical error associated to the physical channel error in order to decide if the decoding round was successful. The operational figure of merit used to evaluate the performance of these quantum error correction schemes is the Word Error Rate ($\mathrm{WER}$), which is the probability that at least one qubit of the received block is incorrectly decoded.

The number of transmitted blocks $N_{\mathrm{blocks}}$ needed to empirically estimate $\mathrm{WER}$ (by Monte Carlo simulation), is given by the following rule of thumb \cite{MonteCarlo}:
\begin{equation}
N_{\mathrm{blocks}} = \frac{100}{\mathrm{WER}}.
\end{equation}
Under the assumption that the observed error events are independent, the above number of blocks will guarantee that the unknown value of $\mathrm{WER}$ will be inside the confident interval $(0.8\hat{\mathrm{WER}} , 1.25\hat{\mathrm{WER}})$ with probability 0.95, where $\hat{\mathrm{WER}}$ refers to the empirically estimated value of $\mathrm{WER}$ based on $N_{\mathrm{blocks}}$.

Figure \ref{fig:QTCTV} shows the performance curves for the quantum turbo codes studied in \cite{josu3,TVQC} operating over the static and TV channels. Here we use $p$ instead of $\gamma$, since the depolarizing channel is considered (calculating the depolarizing probability of a CTA associated to a value of $\gamma$ is trivial). This code has rate $R_\mathrm{Q}=1/9$, a block length of $1000$ qubits, and is decoded using the turbo decoding algorithm presented in \cite{QTC,EAQTC}, which combines two Soft-In Soft-Out (SISO) decoders. From this figure, it can be observed that the performance degradation for the QTC when the channel is the TVADCTA rather than the static ADCTA begins at the waterfall region and becomes more prominent as the depolarizing probability decreases. For example, for a depolarizing probability $p\approx 0.12$, the Word Error Rate, $\mathrm{WER}$, of the QTC operating over the static channel is within the range of $\mathrm{WER} \approx 10^{-2}$, but for that same $p$, it increases an order of magnitude to $\mathrm{WER} \approx 10^{-1}$ when operating over the TV channel scenario QA\_C6. This $\mathrm{WER}$ deviation increases almost four orders of magnitude when the depolarizing probability decreases to $p\approx 0.1$ for the same superconducting qubit scenarios. Thus, as concluded in \cite{TVQC}, the fluctuations of the relaxation time of the superconducting qubits substantially worsen the error correcting capabilites of the QECCs.

\begin{figure}[!h]
\centering
\includegraphics[width=\linewidth]{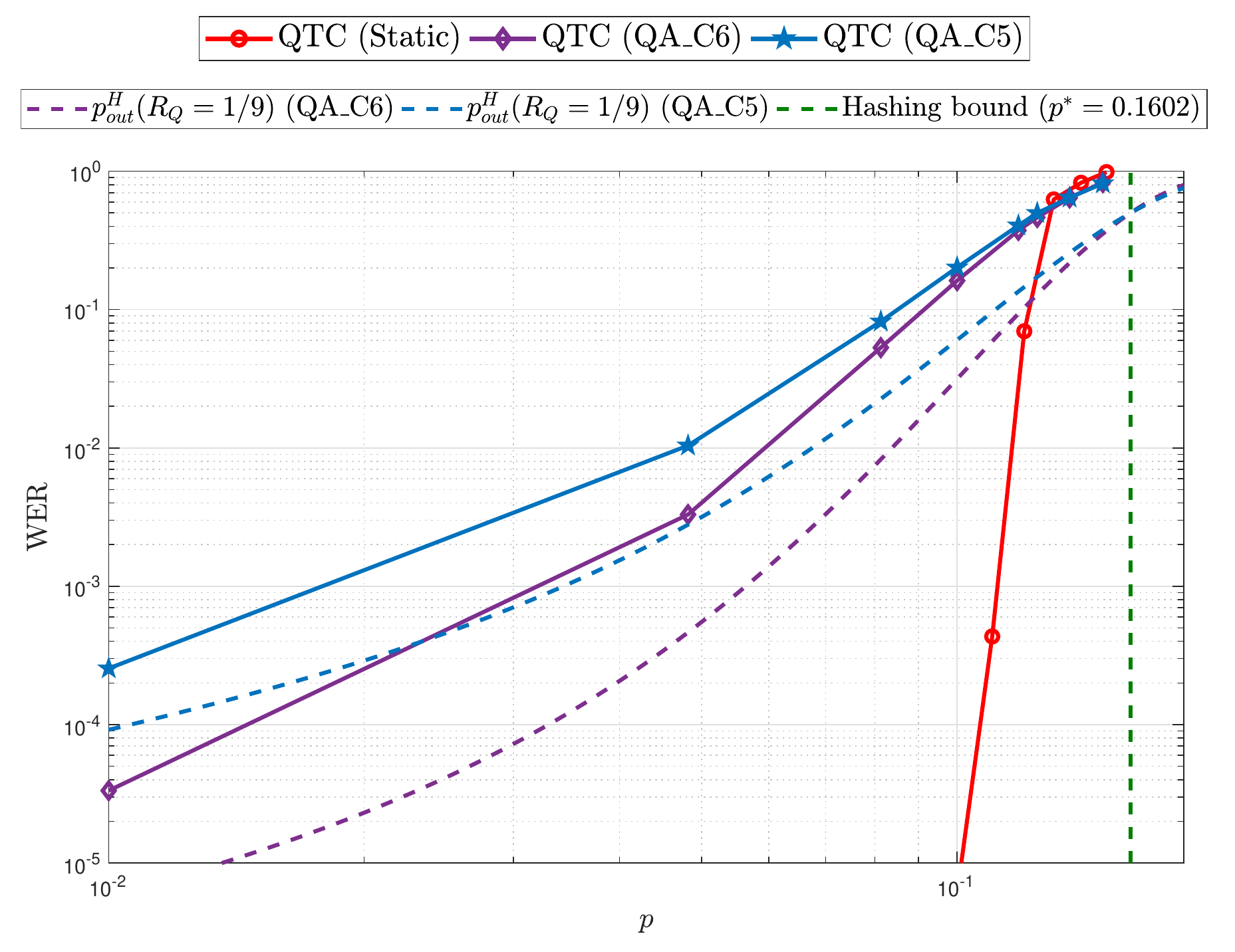}
\caption{\textbf{Performance of the quantum turbo code from \cite{josu3,TVQC}.} The quantum turbo code operates over the static ADCTA channel and the TVADCTA channel for scenarios QA\_C5 and QA\_C6 \cite{decoherenceBenchmarking,TVQC}. The quantum error correction code has rate $R_\mathrm{Q}=1/9$ and encodes blocks of $1000$ qubits. The hashing bound and the hashing outage probabilities are also plotted.}
\label{fig:QTCTV}
\end{figure}

Figure \ref{fig:QTCTV} also shows the quantum hashing outage probability $p_{\mathrm{out}}^\mathrm{H}(R_\mathrm{Q}, p)$, derived in section \ref{sub:TVtwirls}. Note that the quantum hashing outage probability is an upper bound on the asymptotically achievable $\mathrm{WER}$. We know from the previous section that the coefficient of variation of the relaxation time yields insight on how flat the hashing outage probability becomes. Notice that the superconducting qubits of scenario QA\_C5 ($c_\mathrm{v}(T_1)\approx 26\%$) are more affected by time-variations than the ones of scenario QA\_C6 ($c_\mathrm{v}(T_1)\approx 22\%$).

To quantify the distance to the hashing outage bound, we use a similar metric to the one proposed in \cite{patrick3}, which measures the distance in dBs between the performance of a code and the hashing outage at a given $\mathrm{WER} = \chi$:
\begin{equation}\label{eq:distancPout}
\delta_{\mathrm{out}}(@\chi) = 10\log\left(\frac{p(p_{\mathrm{out}}^\mathrm{H}=\chi)}{p(\mathrm{WER}_{\mathrm{code}} = \chi)}\right).
\end{equation}

For example, the $1/9$-QTC is $\delta_{\mathrm{out}}(@10^{-3}) \approx 1.75 \text{ dB}$ away from the hashing outage for the QA\_C5 scenario and $\delta_{\mathrm{out}}(@10^{-3}) \approx 1.67 \text{ dB}$ away for the QA\_C6 scenario.

\section{Conclusion}\label{sec:conc}

In this paper, we have introduced the concept of quantum outage probability as the asymptotically achievable error rate for quantum error correction when time-varying quantum channels are considered. Additionally, we have also introduced the quantum hashing outage probability as an upper bound on the quantum outage probability when Pauli channels are considered, since the actual quantum capacity of these channels is not known. We have provided closed-form expressions of these probabilities for the TVAD, TVADPTA and the TVADCTA channels. We have also studied the behaviour of the $R_\mathrm{Q}=1/9$ QTC from \cite{josu3,TVQC} and benchmarked its performance using the hashing outage probability. We have concluded that the time-variations experienced by the relaxation times do affect the performance of QECCs in a significant manner when the error rate curve is steep enough \cite{TVQC}, and that those TV effects should be taken into account when optimizing code construction. The information theoretical analysis presented in this work is essential to benchmark the behaviour of quantum error correction codes in TV scenarios. Similar studies for the quantum outage probability of the more general time-varying combined amplitude and phase damping channel will be considered in future work, as this will be critical in order to have a complete tableau of the information theoretical limits of error correction for superconducting qubits with pure dephasing channels. 

In summary, it is clear that the time-varying nature of the decoherence parameters will have a significant impact on the performance of future QECCs that will be used to protect quantum information. We have found, based on the results shown throughout this paper, that the quantum outage probability is a function of the coefficient of variation of $T_1$ and that it increases as $c_\mathrm{v}(T_1)$ increases. Therefore, to improve the error correction capabilities of quantum codes in TVQCs it is important to experimentally look for qubits that, not only have a large mean relaxation time, but that also exhibit a low standard deviation relaxation time. In this way, the error correction potential of QECCs under time varying conditions will approach those found for static channels.

\section*{Acknowledgements}
This work was supported by the Spanish Ministry of Economy and Competitiveness through the ADELE project (PID2019-104958RB-C44). This work has been funded in part by NSF Award CCF-2007689. Josu Etxezarreta is funded by the Basque Government predoctoral research grant.

\section*{Author Contributions}
J.E.M. and J.G.-F. conceived the research. J.E.M. derived the closed-form expressions of the Theorem and the Corollary. J.E.M. and P.F. performed the numerical simulations. J.E.M. and P.F. analyzed the results and drew the conclusions. The manuscript was written by J.E.M. and P.F., and revised by P.M.C. and J.G.-F. The project was supervised by P.M.C. and J.G.-F.

\end{document}